\documentclass[letterpaper, 10 pt, conference]{ieeeconf}  

\IEEEoverridecommandlockouts                              

\usepackage{xcolor}
\usepackage{amsmath}
\usepackage{amsfonts}
\usepackage{graphicx}
\usepackage{xparse}
\usepackage{mathtools}
\usepackage{stmaryrd}
\usepackage[ruled,vlined]{algorithm2e}
\usepackage{caption}
\usepackage{subfigure}
\usepackage{subcaption}
\usepackage{tikz}
\usepackage{graphicx}
\usepackage{svg}
\usepackage{tikz}
\usepackage{amsmath}
\usepackage{amssymb}
\usepackage{algorithmic}
\usepackage{mathabx}
\usepackage{hyperref}
\usepackage{lipsum}
\usepackage{booktabs}
\usepackage[most]{tcolorbox}
\hypersetup{
    colorlinks=true,
    linkcolor=black,
    urlcolor=blue,
    pdftitle={},
    pdfpagemode=FullScreen,
}
\usepackage{diagbox}
\usepackage{placeins}

\usetikzlibrary{shapes.geometric, arrows}
\usetikzlibrary{calc,fit,arrows}

\newtheorem{theorem}{Theorem}[section]

\newtheorem{assumption}[theorem]{Assumption}

\newtheorem{proposition}[theorem]{Proposition}
\newtheorem{definition}[theorem]{Definition}
\newtheorem{lemma}[theorem]{Lemma}
\newtheorem{remark}[theorem]{Remark}

\newtheorem{example}{Example}

\newcommand{\R}{{\mathbb{R}}}
\newcommand{\N}{{\mathbb{N}}}

\newcommand{\dom}{\textrm{dom}}

\newcommand{\cl}{\textrm{cl}}
\newcommand{\col}{\operatorname{col}}

\newcommand{\intt}{\textrm{int}}

\title{\LARGE \bf
Invariance is Compositional for Continuous-time Systems: From Sleekness to Lebesgue Density}

\author{Othman Cherkaoui Dekkaki, Sadek Belamfedel Alaoui, Olayo Reynaud,\\ Mohamed Maghenem, Alessio Iovine and Adnane Saoud
\thanks{Othman Cherkaoui Dekkaki, Sadek Belamfedel Alaoui and Adnane Saoud are with College of Computing, University Mohammed VI Polytechnic (UM6P), Benguerir, Morocco. Olayo Reynaud and Mohamed Maghenem are with Univ.\ Grenoble Alpes, CNRS, Grenoble INP, GIPSA-lab, 38000 Grenoble, France. Alessio Iovine is with the Laboratory of Signals and Systems (L2S), CNRS, CentraleSupélec, Paris-Saclay University, Gif-sur-Yvette, France.}
\thanks{*This work was supported by the ENSUS Project under the College of Computing, Mohammed VI Polytechnic University. It has received funding from the Agence Nationale de la Recherche (ANR) via grant ESTHER ANR-22-CE05-0016.}
}

\begin{document}

\maketitle
\thispagestyle{empty}
\pagestyle{empty}


\begin{abstract}
 This work establishes the first bidirectional compositional invariance result: robust forward invariance of a global Cartesian product set is shown to be equivalent to robust forward invariance of each local subsystem under coupling inputs from its neighbors. To facilitate this, we introduce the notion of tangential Lebesgue-density, a new condition that is weaker than classical sleekness but sufficient to ensure that the product of individual tangent cones equals the tangent cone of the product set. This equivalence reduces the curse of dimensionality by allowing the verification of a high-dimensional global system to be decomposed into a series of local sub-checks. The framework's scalability is demonstrated through a DC microgrid numerical example, confirming that the verification complexity grows only linearly with the number of subsystems.
\end{abstract}

\section{Introduction}

Safety of a dynamical system is the property that the state trajectories remain within a prescribed set of admissible states for all time. This property is formalized through the notion of forward invariance: a set $K$ is forward invariant for a given dynamics if every trajectory starting in $K$ remains in $K$ for all future time~\cite{blanchini2008set,aubin2009viability}. For a single continuous-time system described by a differential inclusion, a closed set $K$ is forward invariant if and only if the velocity field points into $K$ at every boundary point; that is, it belongs to the contingent cone to $K$ at each point of $\partial K$~\cite{aubin2009viability}. While this characterization is well understood for monolithic systems, verifying it for 
large-scale interconnected systems is computationally intractable in general. Indeed, the boundary of the global safe set lives in a space whose dimension grows with the number of subsystems, and the tangent cone condition must be checked at every point of this boundary, making the condition infeasible at scale.
 
To overcome the curse of dimensionality inherent to centralized approaches,
a natural strategy is to decompose the verification problem into smaller
subproblems, one per subsystem, and to derive global invariance guarantees
from local ones. A prominent framework for this is that of assume-guarantee
contracts, in which, each subsystem is assigned a contract specifying the
property it must fulfill under some assumptions on the behavior of its
environment. Compositionality results then establish that if all subsystems
satisfy their local contracts, a global contract for the interconnected
system is satisfied. This framework has been developed for discrete-time \cite{Ghasemi2020HSCC}, continuous-time \cite{SaoudAutomatica21} and hybrid systems \cite{Alaoui2024}, covering a broad range of specifications. A general assume-guarantee framework to achieve forward invariance, covering both discrete and continuous-time settings
as well as sampled-data systems, is developed
in~\cite{SaoudECC18,SaoudAutomatica21,SaoudTAC21}. This framework was
extended to signal temporal logic specifications for continuous-time
nonlinear interconnected systems in~\cite{Liu2025TAC} and to behavioral contracts for continuous-time linear dynamical systems in~\cite{Shali2021ECC}. In the context of dissipativity theory, a compositional approach to safety verification for interconnected continuous-time systems was proposed in~\cite{Coogan2015}. 
 
A common feature of all the works presented above is that the compositional reasoning establishes invariance only in the sufficient direction: if each subsystem satisfies its local contract, or maintains its trajectories within a local set, then the global system is safe. The reverse implication, that global invariance of the interconnected system necessitates local invariance of each subsystem, is not addressed in any of these works. To the best of our knowledge, no existing work in the literature establishes an if-and-only-if compositional invariance result for continuous-time interconnected systems.

The present paper addresses this gap and makes two main contributions. The first and main contribution is Theorem~\ref{thm:GBL}, which establishes the first if-and-only-if compositional invariance result for continuous-time interconnected systems: Robust forward invariance of the global Cartesian product set for the interconnected system is equivalent to robust forward invariance of each local set for the corresponding subsystem. This equivalence reduces the verification of a global invariance property, which requires checking a tangent cone condition at every point of the boundary of the global safe set, to a collection of local conditions, each involving only the dynamics of a single subsystem and the geometry of a single local set. The second contribution is Proposition~\ref{lem:product_cone}, which relaxes the classical sleekness assumption required for the equality between the contingent cone to the Cartesian-product set and the product of contingent cones to the individual sets \cite{aubin2009viability,RockafellarWets1998}. We introduce tangential Lebesgue-density, a condition strictly weaker than sleekness. 
It requires that, at every boundary point, the set of admissible times associated with near-tangent directions has Lebesgue density one. By ensuring that valid times occupy almost the entire interval near zero, this property provides the necessary density to synchronize multiple local sequences into a single global one. We demonstrate that sleekness implies tangential Lebesgue-density, whereas the converse is generally false. Furthermore, we prove that tangential Lebesgue-density is sufficient to establish the property that the Cartesian product of local tangent cones is contained within the tangent cone of the global product set. This result provides the theoretical foundation for the bidirectional compositional invariance, allowing global safety to be equivalent to the robust forward invariance of individual subsystems.
 
In the remainder of this paper, Section~\ref{sec:prelim} introduces the interconnection framework, robust forward invariance, and the considered tangent cones.
Section~\ref{sec:cones} studies the Cartesian product of tangent cones, introduces the tangential Lebesgue-density condition, and establishes its relationship to classical sleekness condition. Section~\ref{sec:comp} presents the main compositional invariance result. Section~\ref{sec:numerical} illustrates the theoretical results on a DC microgrid example with $100$ distributed generation units.

\paragraph*{\textbf{Notation}}
We denote by $\R$ the set of real numbers and by $\N$ the set of nonnegative integers, with standard restrictions (e.g., $\R_{\ge0}$, $\R_{>0}$, $\N_{>0}$). $\|\cdot\|$ denotes the Euclidean norm. For a set $K\subset\R^n$, $\cl(K)$ and $\intt(K)$ denote its closure and interior. For $\eta>0$, $\mathbb{B}(0,\eta)$ denotes the ball centered at the origin of radius $\eta$ w.r.t to the Euclidean norm. 
Given sets $S_1,\ldots,S_N$, their Cartesian product is $\prod_{i=1}^N S_{i}$.
For vectors $v_1,\ldots,v_N$, the column concatenation is $\col(v_1,\ldots,v_N):=[v_1^\top\ \cdots\ v_N^\top]^\top.$ 
For a set-valued map $G:\mathbb{R}^{q}\rightrightarrows
\mathbb{R}^{n}$, its domain is
$\dom(G):=\{z\in\mathbb{R}^{q}\mid G(z)\neq\emptyset\}$. For a single-valued map $h$, $\dom(h)$ denotes the set on which
$h$ is defined. In particular, for a signal
$x:I\to\mathbb{R}^{n}$, where $I\subseteq\mathbb{R}_{\geq 0}$,
we have $\dom(x)=I$. Given a set $S \subseteq \mathbb{R}$, $\lambda(S)$ denotes its Lebesgue measure on $\mathbb{R}$.

\section{Preliminaries}\label{sec:prelim}

\subsection{Interconnected subsystems}

Consider a collection of $N \in \mathbb{N}_{>0}$ subsystems $\Sigma_i$, $i \in I:=\{1,2,\ldots,N\}$. Each subsystem $\Sigma_i$ is a tuple 
\begin{equation}
\label{def:agent}
  \Sigma_i:=(X_i,W_i^{1},W_i^{2},Y_i,f_i,h_i),  
\end{equation}
where
\begin{itemize}
    \item $X_i\subset\R^{n_i}$, $W_i^{1}\subset\R^{m_i^{1}}$, $W_i^{2}\subset\R^{m_i^{2}}$ and $Y_i\subset\R^{p_i}$ are the state, external input, internal input and output spaces, respectively.
    
\item $f_i:\R^{n_i}\times\R^{m_i^{1}}\times\R^{m_i^{2}}\rightrightarrows\R^{n_i}$ represents the dynamics of the subsystem and $h_i:\R^{n_i}\to\R^{p_i}$ the output map.
\end{itemize}
A trajectory of subsystem $\Sigma_i$ is a quadruple $(w_i^{1},w_i^{2},x_i,y_i):\dom(x_i)\rightarrow W^1_i \times W^2_i \times X_i \times Y_i$, where $\dom(x_i)\in \{[0,T]\mid T\in\R_{\ge0}\}\ \cup\ \{[0,T)\mid T\in\R_{>0}\}\ \cup\ \R_{\ge0}$, $w^1_i$ and $w^2_i$ are locally measurable, $x_i$ and $y_i$ are absolutely continuous and satisfy for almost all $t \in \dom(x_i)$
\[
\begin{cases}
\dot{x}_i(t) \in f_i(x_i(t), w_i^1(t), w_i^2(t)), \quad  x_i(0) \in X_i, \\
y_i(t) = h_i(x_i(t)).
\end{cases}
\]
We use $\phi_i(t,x^0_i,w_i^1,w_i^2)$ to denote the state of the system $\Sigma_i$ reached at time $t\geq 0$, from the initial condition $x_i(0) \in X_i$, under the internal input $w^1_i:[0,t] \rightarrow W_i^1$ and external input $w^2_i:[0,t] \rightarrow W_i^2$.

\begin{remark}
    The tuple $\Sigma_i$ distinguishes between two types of inputs. The external input $w_i^1 \in W_i^1$ represents exogenous disturbances (outside the network) acting on subsystem $\Sigma_i$, such as environmental perturbations or unmodeled  dynamics. The internal input $w_i^2 \in W_i^2$, by contrast, captures the coupling between $\Sigma_i$ and the rest of the network: It carries the outputs of the neighboring subsystems, and is determined  endogenously once the interconnection structure is fixed. The interconnection structure and the role of internal  inputs in the invariance analysis are made explicit in Definition~\ref{def:mas} and Theorem~\ref{thm:GBL}.
\end{remark}

In the remainder of the paper, we make the following regularity assumption.

\begin{assumption}\label{ass:SA}
Consider the subsystem $\Sigma_i$, $i \in I$, defined above. The set-valued map $f_i$ is Lipschitz\footnote{
the set-valued map $f_i$ is locally Lipschitz if for every $z\in\intt(\dom(f_i))$ there exist a neighborhood $U$ of $z$
and $L\ge0$ such that $f_i(z_1)\subseteq f_i(z_2)+L\|z_1-z_2\|\mathbb{B}(0,1)$ for all $z_1,z_2\in U\cap\dom(f_i)$.
It is Lipschitz if $L$ can be chosen independent of $z\in\intt(\dom(G))$.
$f_i$ has (nonempty) compact values if $f_i(z)$ is compact for all $z\in\dom(f_i)$.}
has (nonempty) compact values, and $
X_i\times W_i^{1}\times W_i^{2}\subset \intt\bigl(\dom(f_i)\bigr)$. The map $h_i$ is continuous, $X_i\subset \intt\bigl(\dom(h_i)\bigr)$, and $h_i(X_i)\subset Y_i$.
\end{assumption}

\begin{remark}
Assumption~\ref{ass:SA} consists of standard regularity conditions that are classical in the theory of differential inclusions 
\cite{aubin2009viability}. The Lipschitz property and compactness  of the values of $f_i$ together guarantee existence of solutions,  and are the minimal conditions under which the  tangent-cone-based invariance characterization holds \cite{aubin2009viability}.
\end{remark}

A network of subsystems consists of a collection of subsystems $\Sigma_i$, $i\in I=\{1,2,\ldots,N\}$, and a binary connectivity relation $\mathcal{I} \subseteq I \times I$. For $i \in I$, we define $\mathcal{N}(i):=\{j \in I \mid (j,i)\in \mathcal{I}\}$ as the set of neighbouring subsystems from which the incoming edges originate. We define the linear map $\pi_i : \mathbb{R}^n \to \mathbb{R}^{n_i}$, $i \in I$, as the canonical projection, given by $\pi_i(x) := x_i$ for $x = (x_1,\ldots,x_N)$. By slight abuse of notation, we extend $\pi_i$ to subsets $S \subseteq \mathbb{R}^n$ by defining $
\pi_i(S) := \{\, x_i \in \mathbb{R}^{n_i} \mid x \in S \,\}.$

\begin{definition}\label{def:mas}
Consider a collection of subsystems $\Sigma_i$, $i\in I=\{1,2,\ldots,N\}$ as defined in (\ref{def:agent}) and a binary connectivity relation $\mathcal{I} \subseteq I \times I$. We say that $(\Sigma_i)_{i \in I}$ is compatible for composition with respect to $\mathcal{I}$ if for each $i \in I$, we have 
$\prod_{j \in \mathcal{N}(i)} Y_j \subseteq W^2_i$, 
i.e., the internal input space of $\Sigma_i$ is a superset of the Cartesian product of the output spaces of all the neighbors in $\mathcal{N}(i)$. The interconnected system $\left\langle \left(  \Sigma_i \right)_{i \in I}, \mathcal{I} \right\rangle$ is denoted by $\Sigma$ and given by the tuple $\Sigma= (X,W^{1},Y,F,H)$, where
\begin{itemize}
    \item $X:=\prod_{i\in I}X_i$,  $W^{1}:=\prod_{i\in I}W_i^{1}$ and $Y:=\prod_{i\in I}Y_i$ are the sets of states, external inputs and outputs, respectively,
    \item $F:X\times W^{1}\rightrightarrows\R^{n}$, $n:=\sum_{i\in I}n_i$ is the dynamics of the system and $H:X\to Y$ is the output map, defined for $x=\col(x_i)_{i\in I}\in X$ and $w^{1}=\col(w_i^{1})_{i\in I}\in W^{1}$ as
    \begin{align*}
    F(x,w^{1})
    :=\Bigl\{&\col(v_i)_{i\in I}\ \Big|\ 
    v_i\in f_i\bigl(x_i,w_i^{1},w_i^2), \\ &w_i^2=\col(h_j(x_j))_{j\in \mathcal{N}(i)}\in W_i^2, \ \forall i \in I\Bigl\}
    \end{align*} 
and $H(x):=\col\bigl(h_i(x_i)\bigr)_{i\in I}$.
\end{itemize}
A trajectory of system $\Sigma$ is a tuple $(w^{1},x,y):\dom(x)\rightarrow W^1 \times X \times Y$, where $w^1$ is locally measurable, $x$ and $y$ are absolutely continuous and satisfy for almost all $t \in \dom(x)$
\[
\begin{cases}
\dot{x}(t) \in F(x(t), w^1(t)), \; x(0) \in X, \\
y(t) = H(x(t)).
\end{cases}
\]
\end{definition}

We have the following auxiliary result, showing that Assumption \ref{ass:SA} is directly transferable from subsystems to the interconnected system.

\begin{lemma}\label{lem:stacked-reg}
Consider a network of subsystems, $\Sigma_i$, $i \in I$, compatible for composition with respect to $\mathcal{I}$. If each subsystem $\Sigma_i$ satisfies Assumption \ref{ass:SA}, then the interconnected system $\Sigma=\left\langle \left(  \Sigma_i \right)_{i \in I}, \mathcal{I} \right\rangle$ satisfies Assumption \ref{ass:SA}.
\end{lemma}

\begin{proof}
We verify each condition of Assumption~\ref{ass:SA} for the interconnected system $\Sigma$. For each $i \in I$, since $f_i$ is Lipschitz and has compact values, and since $h_i$ is continuous, the composed map $x \mapsto f_i(x_i, w_i^1, \col(h_j(x_j))_{j \in \mathcal{N}(i)})$ 
is Lipschitz in $x$. Since $F(x,w^1) = \col(f_i(\cdot))_{i \in I}$ 
is a stacking of Lipschitz set-valued maps, $F$ is a Lipschitz set-valued map. Moreover, compact values of $F$ follow from the fact that finite Cartesian products of compact sets are compact.  Regarding the domain condition, since $X_i \times W_i^1 \times W_i^2 \subset \intt(\dom(f_i))$  for each $i \in I$, and since the interconnection sets satisfy  $\prod_{j \in \mathcal{N}(i)} h_j(X_j) \subseteq W_i^2$, it follows that $X \times W^1 \subset \intt(\dom(F))$. Finally, the map $H(x) = \col(h_i(x_i))_{i \in I}$ is continuous as a finite composition of continuous maps, $X \subset \intt(\dom(H))$, and $H(X) \subseteq Y$ since $h_i(X_i) \subseteq Y_i$ for each $i \in I$.
\end{proof}

\subsection{Robust forward invariance}

We now introduce the notion of robust forward-invariance. To do so, given $i \in I$, we consider $\Sigma_i=(X_i,W_i^{1},W_i^{2},Y_i,f_i,h_i)$ defined in (\ref{def:agent}).

\begin{definition}[Robust forward invariance]\label{def:FI}
A set $K_i\subseteq X_i$ is robust forward invariant for $\Sigma_i$ under ($\mathcal{W}^1_i, \mathcal{W}^2_i)  \subseteq W^1_i \times W^2_i$ if every solution  $\phi_i(.,x^0_i,{w}^1_i,{w}^2_i):\dom(\phi_i) \rightarrow X_i$ satisfies $\phi_i(t,x^0_i,{w}^1_i,{w}^2_i) \in K_i$ for all $ t \in \dom(\phi_i)$, for every initial condition $x^0_i \in K_i$, and for all locally measurable external and internal input signals ${w}^1_i:\dom(\phi_i) \rightarrow \mathcal{W}_i^1$ and ${w}^2_i:\dom(\phi_i) \rightarrow \mathcal{W}_i^2$.
\end{definition}

\begin{remark} 
A similar definition holds for the interconnected system  $\Sigma$ in Definition~\ref{def:mas}. That is, a set $K \subseteq X$ is robustly forward invariant for $\Sigma$ under $\mathcal{W}^1 \subseteq W^1$ if every solution $\phi(t, x^0, w^1)$ satisfies $\phi(t, x^0, w^1) \in K$ for all $t \in \dom(\phi)$, for every initial condition $x^0 \in K$, and for all locally measurable external input  signals $w^1 : \dom(\phi) \to \mathcal{W}^1$. Unlike Definition~\ref{def:FI}, no internal input constraint set appears.
\end{remark}

We also recall below the notions of contingent cone, Clarke tangent cone, and sleek sets, which play a central role in our study.

\begin{definition}\label{def:TK}
The contingent cone to the set $K \subset \mathbb{R}^n$ at $x \in \text{cl}(K)$, denoted by $T_K(x)$, is given by
  \[ T_K(x) := \left\{ v \in \mathbb{R}^n : \exists\, t^k \to 0^+,~ 
\exists\, w^k \rightarrow v, ~ x + t^k w^k \in K \right\}.\]
  
The Clarke tangent cone to the set $K \subset \mathbb{R}^n$ at $x \in \text{cl}(K)$, denoted by $C_K(x)$, is given by
  \[
  \begin{split}
 \hspace{-0.2cm} &  C_K(x) := 
  \\ 
   \hspace{-0.2cm}
  &
  \left\{ v \in \mathbb{R}^n : \forall t^k \to 0^+,~\forall x^k {\rightarrow}_K x, ~ \exists w^k \rightarrow v, ~ x^k + t^k w^k \in K \right\},
  \end{split}
  \]
  where $x^k {\rightarrow}_K x$ means that the sequence $\{x^k\}_{k} \subset K$ and $x^k {\rightarrow} x$.
\end{definition}
Note that the map $x \mapsto C_K(x)$ is lower semicontinuous, and $C_K(x)$ is a closed and convex subset of $T_K(x)$.

\begin{definition}
    A closed set $K  \subset \mathbb{R}^n$ is sleek at $x \in K$ if $C_{K}(x) = T_{K}(x)$, or equivalently, the map $z \mapsto T_{K}(z)$ is lower semicontinuous at $x$ \cite{aubin2009viability}. 
    The set $K$ is said to be sleek, if it is sleek at every $x \in K$.
\end{definition}

\section{Cartesian product of tangent cones}\label{sec:cones}

For a Cartesian product set $K = \prod_{i \in I} K_i$, the inclusion 
\begin{equation}
\label{eq:cone_inclusion_trivial}
   T_K(x) \;\subseteq\; \prod_{i \in I} T_{K_i}(x_i), 
\end{equation}
holds for any closed sets $K_i$ and any $x \in \partial K$.
Indeed, if $x + t^k v^k \in K = \prod_i K_i$, then
$x_i + t^k v_i^k \in K_i$ for each $i \in I$, so any tangent direction to $K$ at $x$ projects onto a tangent direction to $K_i$ at $x_i$.
The converse inclusion, 
\begin{equation}
    \prod_{i \in I} T_{K_i}(x_i) \;\subseteq\; T_K(x),
\end{equation}
is more delicate since it requires synchronizing, for each
$v = \col(v_i)_{i \in I} \in \prod_i T_{K_i}(x_i)$, the $N$ approximating sequences provided by the local tangent-cone definition into a single common sequence. Classically, this is achieved under  sleekness of each $K_i$~\cite[Prop. 6.41]{RockafellarWets1998}. In the following proposition, we show that sleekness can be replaced by a strictly weaker, pointwise condition that we introduce as tangential Lebesgue-density. Unlike sleekness, which constrains the geometry of the set $K_i$ in an entire neighborhood of $x_i$, this condition is local to each fixed point $x_i \in \partial K_i$: it only requires that the set of valid approximation times starting from $x_i$ is dense near $0$ in the sense of Lebesgue measure, placing no restrictions on the behavior of $K_i$ near neighboring points.

\begin{definition}[Tangential Lebesgue-density]
\label{def:TLD}
Let $K \subseteq \mathbb{R}^{n}$ be a closed set and let 
$x \in \partial K$. For $v \in T_K(x)$ and $\varepsilon > 0$, 
define the valid time set
\begin{equation}\label{eq:Sdef}
\begin{aligned}
    S_K(x,v,\varepsilon) \;:=\;
    \bigl\{t > 0 :&\; \exists\, w \in \mathbb{R}^{n},\;
    \|w - v\|<\varepsilon, 
    \\
    & \; x + tw \in K \bigr\}.
\end{aligned}
\end{equation}
The set $K$ is called tangentially Lebesgue-dense at $x$ if
\begin{equation}\label{eq:FDC}
    \lim_{r \to 0^+}
\frac{\lambda\bigl(S_K(x,v,\varepsilon)\cap(0,r)\bigr)}{r}
    = 1, ~~ \forall\, v \in T_{K}(x),\; \forall\, \varepsilon>0,
\end{equation}
where $\lambda$ denotes the Lebesgue measure on $\mathbb{R}$.
\end{definition}

Intuitively, the set $S_K(x,v,\varepsilon)$ collects all times $t > 0$ for which there exists a vector $w$ within distance $\varepsilon$ of $v$ such that $x + tw \in K$. For small $\varepsilon$, the vector $w$ is close to $v$ both in magnitude and orientation, so $S_K(x,v,\varepsilon)$ captures the times at which $x$ can be displaced by approximately $tv$ while staying in $K$. The ratio $\lambda(S_K(x,v,\varepsilon)\cap(0,r))/r$ measures the proportion of the time window $(0,r)$ that consists of valid approximation times. Condition~\eqref{eq:FDC} requires this proportion to tend to $1$ as $r \to 0^+$: for small $r$, almost the entire window $(0,r)$ consists of valid times, and only a negligible fraction, in the sense of length, fails this property. This is strictly stronger than the mere existence of valid times arbitrarily close to $0$, which is all that $v \in T_K(x)$ guarantees, but strictly weaker than the classical sleekness property, as established in Proposition~\ref{prop:sleek_implies_FDC}.

\begin{proposition}
    \label{lem:product_cone}
Let $K_i \subseteq \mathbb{R}^{n_i}$, $i \in I = \{1,\ldots,N\}$, be closed sets and define $K := \prod_{i \in I} K_i \subseteq \mathbb{R}^n$. Suppose that for all $i \in I$, the set $K_i$ is tangentially Lebesgue-dense at every $x_i \in \partial K_i$ in the sense of Definition~\ref{def:TLD}. Then, for all $x \in \partial K$,
\[
    \prod_{i \in I} T_{K_i}(x_i) \;\subseteq\; T_K(x).
\]
\end{proposition}

\begin{proof}
Let $v = \col(v_i)_{i \in I} \in \prod_{i \in I} T_{K_i}(x_i)$. Hence, $v_i \in T_{K_i}(x_i)$ for every $i \in I$. To show the result, we need to construct sequences $t^k \to 0^+$ and $w^k \to v$ with
$x + t^k w^k \in K$ for all $k \in \mathbb{N}$.

We first start by constructing the sequence of $t^k$. Consider $k \in \mathbb{N}$ and $i \in I$ and the set $S_{K_i}(x_i,v_i,1/k)$ defined according to Definition~\ref{def:TLD} with $\varepsilon = 1/k$. Since $S_{K_i}(x_i,v_i,1/k)\cap(0,r)\subseteq (0,r)$, the ratio
$\lambda(S_{K_i}(x_i,v_i,1/k)\cap(0,r))/r$ lies in $[0,1]$ for all $r > 0$, where $\lambda$ is the Lebesgue measure on $\mathbb{R}$.  Since the ratio lies in $[0,1]$ and
converges to $1$, one gets from (\ref{eq:FDC}) the existence of $r_i^k > 0$ depending on $i$ and $k$ such that for all
$r \in (0, r_i^k)$:
\begin{equation}\label{eq:FDC_applied}
    1 - \frac{1}{N}
    \;<\;
    \frac{\lambda\bigl(S_{K_i}(x_i,v_i,1/k)\cap(0,r)\bigr)}{r}
    \;\leq\; 1.
\end{equation}
where $N$ is cardinal of the set of subsystems $I$. Since $I$ is finite, define $\bar{r}^k := \min_{i \in I} r_i^k > 0$. Hence, using \eqref{eq:FDC_applied} one has that 
\begin{equation}
\label{eqn:Prop2}
    1 - \frac{1}{N}
    \;<\;
    \frac{\lambda\bigl(S_{K_i}(x_i,v_i,1/k)\cap(0,r)\bigr)}{r}
    \;\leq\; 1,
\end{equation}
all $i \in I$ and all $r \in (0,\bar{r}^k)$. Pick any $r^k \in (0, \min(\bar{r}^k, 1/k))$, and let us show that,
\begin{equation}\label{eq:intersection_nonempty}
    \bigcap_{i \in I} S_{K_i}\!\left(x_i,v_i,\frac{1}{k}\right)
    \cap\, (0, r^k) \;\neq\; \emptyset.
\end{equation}
Suppose for contradiction this intersection is empty. Hence, for all $t \in (0,r^k)$ there exists $i \in I$ such that $t \in (0,r^k) \setminus S_{K_i}(x_i,v_i,1/k)$. Hence, one gets that 
\begin{equation}
\label{eqn:Prop21}
    (0,r^k) \subseteq \bigcup_{i \in I}
((0,r^k)\setminus S_{K_i}(x_i,v_i,1/k)).
\end{equation}
By applying the subadditivity of $\lambda$ to (\ref{eqn:Prop21}), one gets for $r = r^k$ (since $r \in (0,\bar{r}^k)$ and $r^k< \bar{r}^k$) that
\begin{align*}
r^k &\leq \sum_{i \in I}
\lambda\bigl((0,r^k)\setminus S_{K_i}(x_i,v_i,1/k)\bigr)\\
&= \sum_{i \in I}\!\Bigl(r^k -
\lambda\bigl(S_{K_i}(x_i,v_i,1/k)\cap(0,r^k)\bigr)\Bigr)\\
&< \sum_{i \in I}\!\left(r^k
- \Bigl(1-\frac{1}{N}\Bigr)r^k\right)
= \sum_{i \in I}\frac{r^k}{N} = r^k,
\end{align*}
where the third inequality follows from (\ref{eqn:Prop2}), which contradits our assumption. Hence~\eqref{eq:intersection_nonempty} holds
and we choose the sequence $t^k$ satisfying $t^k \in \bigcap_{i \in I} S_{K_i}(x_i,v_i,1/k)
\cap (0,r^k)$ for all $k \in \mathbb{N}$.

We now construct the sequence $w^k$ as follows. Since $t^k \in S_{K_i}(x_i,v_i,1/k)$ for all $i \in I$ and for all $k \in \mathbb{N}$, one has from definition of the sets $S_{K_i}(x_i,v_i,1/k)$, $i \in I$, in \eqref{eq:Sdef} the existence of $w_i^k \in \mathbb{R}^{n_i}$ with
\begin{equation}\label{eq:local_approx}
    \|w_i^k - v_i\| < \frac{1}{k}
    \quad\text{and}\quad
    x_i + t^k w_i^k \in K_i.
\end{equation}
Hence, we define $w^k := \col(w_i^k)_{i \in I}$. 

Let us now show the convergence of the sequences $t_k$ and $w^k$ to $0^+$ and $v$, respectively. First, from the construction of the sequence $r^k \in (0, \min(\bar{r}^k, 1/k))$, one has $t^k < r^k < 1/k$, which implies that $t^k$ converges to $0^+$. Moreover, one has from (\ref{eq:local_approx}) that $\|w^k - v\|^2 = \sum_{i \in I} \|w_i^k - v_i\|^2 < N/k^2
\to 0$, which implies that $w^k$ converges to $v$. Hence, it follows that $x + t^k w^k = \col(x_i + t^k w_i^k) \in
\prod_{i \in I} K_i = K$ for all $k \in \mathbb{N}$. By Definition~\ref{def:TK}, $v \in T_K(x)$. Hence, we conclude that $\prod_{i\in I} T_{K_i}(x_i) \subseteq T_K(x)$.
\end{proof}

\begin{remark}
Condition $v_i \in T_{K_i}(x_i)$ guarantees only that  $S_{K_i}(x_i,v_i,\varepsilon)$ is non-empty when $\varepsilon$ is arbitrarily close to $0$, i.e., valid times exist near $0$ but could be isolated, leaving large gaps in between. The core difficulty in the proof is that each subsystem  provides its own sequence of valid times, and these $N$ sequences are generally different, so they cannot directly serve as a common time for the product. The proposed tangential Lebesgue-density condition in Definition \ref{def:TLD} resolves this by requiring that valid times fill almost the entire interval $(0,r)$ for small $r$, guaranteeing that a common valid time for all subsets must exist.
\end{remark}

The following result establishes the relationship between the tangential Lebesgue-density condition introduced in Definition~\ref{def:TLD} and the classical sleekness assumption. It shows that sleekness implies tangential Lebesgue-density, but that the converse fails in general, confirming that Definition~\ref{def:TLD} is a strictly weaker requirement than sleekness.

\begin{proposition}
\label{prop:sleek_implies_FDC}
Let $K \subseteq \mathbb{R}^{n}$ be a closed set and
$x \in \partial K$. If $K$ is 
   sleek at $x$, then $K$ is tangentially Lebesgue-dense at $x$ in the sense of Definition~\ref{def:TLD}.
\end{proposition}

\begin{proof}
Consider a closed and sleek set $K$, $x \in \partial K$, and pick any $v \in T_K(x) = C_K(x)$ and $\varepsilon > 0$. To show the result, we show the existence of $\delta>0$ such that $(0,\delta) \subseteq S_K(x,v,\varepsilon)$.
The definition of the Clarke tangent cone to the set $K$ at $x$ implies that for any sequence $z^k$ converging to $x$ such that $z^k \in K$ for all $k \in \mathbb{N}$, and for any time sequence $t^k$, $k \in \mathbb{N}$ converging to $0^+$, there exists a sequence $w^k$, $k\in \mathbb{N}$ converging to $v$ such that $z^k + t^k w^k \in K$ for all $k \in \mathbb{N}$.

To show the result, we actually prove the stronger statement that there exists $\delta > 0$ such that  $(0,\delta) \subseteq S_K(x,v,\varepsilon)$, which immediately implies~\eqref{eq:FDC} since $\lambda(S_K(x,v,\varepsilon)\cap(0,r)) = r$ for all $r \in (0,\delta)$. Suppose for contradiction that $S_K(x,v,\varepsilon)$  contains no interval $(0,\delta)$ for any $\delta > 0$. This means that for every $\delta > 0$, the interval $(0,\delta)$ is not entirely contained in $S_K(x,v,\varepsilon)$, i.e., there exists at least one time in $(0,\delta)$ that does not belong to $S_K(x,v,\varepsilon)$. Applying this with $\delta = 1/k$ for each $k \geq 1$, there exists $s^k \in (0, 1/k)$ such that $s^k \notin S_K(x,v,\varepsilon)$. Since $s^k \in (0,1/k)$ for every $k \geq 1$, we have $s^k \to 0^+$.
Applying the Clarke cone condition with the constant sequence $z^k := x$ (which satisfies $z^k \in K$ for all $k$ since $x \in \partial K \subseteq K$, as $K$ is closed) and the time sequence $t^k := s^k \to 0^+$, there exists a sequence $w^k \to v$ such that $x + s^k w^k \in K$ for all $k$. Since $w^k \to v$, there exists $K_0 \geq 1$ such that $\|w^k - v\| < \varepsilon$ for all $k \geq K_0$. Therefore, for all $k \geq K_0$, both $\|w^k - v\| < \varepsilon$ and $x + s^k w^k \in K$ hold simultaneously, which means by definition of $S_K(x,v,\varepsilon)$ in~\eqref{eq:Sdef} that $s^k \in S_K(x,v,\varepsilon)$ for all $k\geq K_0$, contradicting $s^k \notin S_K(x,v,\varepsilon)$.

Hence $(0,\delta) \subseteq S_K(x,v,\varepsilon)$ for some $\delta > 0$, which gives $\lambda(S_K(x,v,\varepsilon)\cap(0,r)) = r$ for all $r \in (0,\delta)$, and therefore $K$ is tangentially Lebesgue-dense at $x$.
\end{proof}
The converse of Proposition  \ref{prop:sleek_implies_FDC} does not generally hold. The following counterexample shows that there exists a closed set $K$ and a point $x$ on its boundary such that $K$ is tangentially Lebesgue-dense at $x$, yet $K$ is not sleek at $x$.
\begin{example}
 Consider the closed set
\[
    K := \bigl\{(x,y) \in \mathbb{R}^2 : y \geq 0\bigr\}
    \;\cup\;
    \bigl\{(0,y) : y \leq 0\bigr\}.
    \]

It is clear that $T_K(0,0) = \{(a,b): b \geq 0\} \cup \{(0,b): b \leq 0\}$. First, let us show that $K$ is tangentially Lebesgue-dense at $z=(0,0)$. Consider $v \in T_K(z)$ and $\varepsilon>0$, we have two cases,
\begin{itemize}
    \item If $v = (a,b)$ with $b \geq 0$, then by taking $w = v$ one gets  $\|w - v\| = 0 < \varepsilon$ and $(0,0) + tw = (ta,tb) \in K$ for all $t > 0$ since $tb \geq 0$. Hence, $\mathbb{R}_{>0} \subseteq S_K(z,v,\varepsilon)$.
\item If $v = (0,b)$ with $b < 0$, then by taking $w = v = (0,b)$, one gets  $\|w - v\| = 0 < \varepsilon$ and $(0,0) + tw = (0,tb) \in K$ for all $t > 0$ since $tb \leq 0$. Hence  $\mathbb{R}_{>0} \subseteq S_K(z,v,\varepsilon)$. 
\end{itemize}
Since $S_K(z,v,\varepsilon) \subseteq \mathbb{R}_{>0}$ and since in both cases we have $\mathbb{R}_{>0} \subseteq S_K(z,v,\varepsilon)$, it follows that $S_K(z,v,\varepsilon) = \mathbb{R}_{>0}$ and the ratio $\lambda(S_K(z,v,\varepsilon)\cap(0,r))/r = 1$ for all $r > 0$, which in turn implies that $K$ is tangentially Lebesgue-dense at $(0,0)$.

Let us now show that $K$ is not sleek at $(0,0)$. Since $C_K(0,0)$ is convex by definition and  $T_K(0,0) = \{(a,b):b\geq 0\}\cup\{(0,b):b\leq 0\}$ is nonconvex, as  $(1,0)\in T_K(0,0)$ and $(0,-1)\in T_K(0,0)$ yet their convex combination $(\tfrac{1}{2},-\tfrac{1}{2})\notin T_K(0,0)$, it follows that $C_K(0,0)\neq T_K(0,0)$, and therefore $K$ is not sleek at $(0,0)$.

\end{example}

\section{Compositional invariance}\label{sec:comp}

We now provide the main result of the paper, by showing that the invariance property is a compositional property.

\begin{theorem}
   \label{thm:GBL}
Consider a network of subsystems, $\Sigma_i$, $i \in I$, compatible for composition with respect to $\mathcal{I}$ and let $\Sigma=\left\langle \left(  \Sigma_i \right)_{i \in I}, \mathcal{I} \right\rangle$ be the interconnected system. For each $i \in I$, let $K_i \subseteq X_i$ and $\mathcal{W}^1_i \subseteq W^1_i$ be closed sets and define $K:=\ \prod_{i\in I} K_i$ and $\mathcal{W}^1:=\ \prod_{i\in I} \mathcal{W}^1_i$. If each subsystem $\Sigma_i$ satisfies Assumption~\ref{ass:SA} and if the sets $K_i$ are tangentially Lebesgue-dense at every $x_i \in \partial K_i$ in the sense of Definition~\ref{def:TLD}, for all $i \in I$, then the following properties are equivalent:
\begin{enumerate}
\renewcommand{\labelenumi}{(\alph{enumi})}
    \item The set $K\subseteq X$ is a robust forward invariant set for the interconnected system $\Sigma$ under the constraint set $\mathcal{W}^1$,
    \item For all $x\in\partial K$, $f(x,\mathcal{W}^1)\subseteq T_{K}(x)$,
    \item For all $i \in I$, the set $K_i\subseteq X_i$ is a robust forward invariant set for the subsystem $\Sigma_i$ under the constraint set ($\mathcal{W}^1_i, \mathcal{W}^2_i)$, where $\mathcal{W}^2_i=\prod_{j\in \mathcal{N}(i)} h_j(K_j)$, 
    \item For all $i\in I$ and for all $x_i\in\partial K_i$,  \[
f_i(x_i,\mathcal{W}^1_i,\mathcal{W}^2_i)\subseteq T_{K_i}(x_i),
\]
where $\mathcal{W}^2_i:=\prod_{j\in \mathcal{N}(i)} h_j(K_j)$.
\end{enumerate}
\end{theorem}

\begin{proof}
To show the result we proceed in four steps.

\noindent\textbf{Step 1: (a) $\Leftrightarrow$ (b).}
Since each subsystem $\Sigma_i$, $i \in I$, satisfies Assumption \ref{ass:SA}, it follows from Lemma \ref{lem:stacked-reg} that the interconnected system $\Sigma$ satisfies Assumption \ref{ass:SA}. Moreover, one has from the compactness of the set $\mathcal{W}^1$ that the set-valued map $f(.,\mathcal{W}^1)$ is Lipschitz and has compact values. Finally, since $K \subseteq X \subseteq \dom(f(.,\mathcal{W}^1))$, it follows from Theorem 5.3.4 in \cite{aubin2009viability} that (a) $\Leftrightarrow$ (b).

\medskip
\noindent\textbf{Step 2: (c) $\Leftrightarrow$ (d).}
Consider the subsystem $\Sigma_i$, $i \in I$. Since $\Sigma_i$ satisfies Assumption \ref{ass:SA}, one has from the compactness of the sets $K_j$, $j \in \mathcal{N}(i)$, and the continuity of the map $h_i$, that the set $\mathcal{W}^2_i=\prod_{j\in \mathcal{N}(i)} h_j(K_j)$ is compact. Moreover, from Assumption \ref{ass:SA} and the compactness of the set $\mathcal{W}^1_i$ one gets that the set-valued map $f_i(.,\mathcal{W}^1_i,\mathcal{W}^2_i)$ is Lipschitz and has compact values. Hence, since $K_i \subseteq X_i \subseteq \dom(f_i(.,\mathcal{W}^1_i,\mathcal{W}^2_i))$, it follows from Theorem 5.3.4 in \cite{aubin2009viability} that (c) $\Leftrightarrow$ (d).

\medskip
\noindent\textbf{Step 3: (b) $\Rightarrow$ (d).}
Consider $i\in I$, $x_i\in\partial K_i$ and $w_i^{2}\in \mathcal{W}_i^{2}$. Since $\mathcal{W}^2_i:=\prod_{j\in \mathcal{N}(i)} h_j(K_j)$, there exists $\hat x_j\in K_j$ such that $w_i^{2}=\col(h_j(\hat{x}_j))_{j\in \mathcal{N}(i)}$. Now we define $z\in K$ by $z_i:=x_i$, $z_j:=\hat x_j$ for all $j \in \mathcal{N}(i)$ and $z_l \in K_l$ for all $l \in I \setminus \left(\mathcal{N}(i) \bigcup \{i\}\right)$. Since $x_i\in\partial K_i$, it follows that $z\in\partial K$. 
Now consider any $w_i^{1}\in \mathcal{W}_i^{1}$ and choose  $w^{1}\in \mathcal{W}^{1}$ whose $i$-th component equals $w_i^{1}$. Hence, one gets from (b) that $f(z,w^1) \subseteq T_K(z)$. Moreover using the inclusion $T_K(z) \subseteq \prod_{l \in I} T_{K_l}(z_l)$ from \eqref{eq:cone_inclusion_trivial}, and the fact that $z_i = x_i$, we obtain $\pi_i(T_K(z)) \subseteq T_{K_i}(z_i) = T_{K_i}(x_i)$. Since $F(z,w^1)$ is the column-stacking of $f_l(z_l, w_l^1, w_l^2)$ over $l \in I$ and the canonical projection $\pi_i$ is compatible with this stacking, we conclude:
    \[
    f_i(x_i,w_i^{1},w_i^{2})
    \;\subseteq\;
    \pi_i\bigl(f(z,w^1)\bigr)
    \;\subseteq\;
    \pi_i\bigl(T_{K}(z)\bigr)
    \;\subseteq\;
    T_{K_i}(x_i).
    \]


\medskip
\noindent\textbf{Step 4: (d) $\Rightarrow$ (b).}
Assume each $K_i$ is tangentially Lebesgue-dense at every  $x_i \in \partial K_i$ and that (d) holds.
Take any $x\in \partial K$ and define $w^2:=(w^2_1,\ldots,w^2_N)$ with $w_i^{2}=\col(h_j(x_j))_{j\in \mathcal{N}(i)} \in \mathcal{W}^2_i$. Now consider any $w^{1}\in \mathcal{W}^{1}$, with $w^1:=(w^1_1,\ldots,w^1_N)$. We have two cases: if $x_i \in \partial K_i$ it follows from (d) that
$f_i(x_i,w_i^{1},w_i^{2})\subseteq T_{K_i}(x_i)$, and if $x_i \in \intt(K_i)$, then it follows again that $f_i(x_i,w_i^{1},w_i^{2})\subseteq T_{K_i}(x_i)=\R^{n_i}$. Hence, one gets
\begin{align*}
 F(x,w^{1})
=&\Big\{\col(v_i)_{i\in I}\ \Big|\ v_i\in f_i(x_i,w_i^{1},w_i^{2})\ \forall i \in I\Big\}
\\ &\subseteq \prod_{i\in I}T_{K_i}(x_i) \subseteq T_K(x),   
\end{align*}
where the last inclusion follows from  Proposition~\ref{lem:product_cone} under the tangential Lebesgue-density of the sets $K_i$, $i\in I$, which concludes the proof.
\end{proof} 

\begin{remark} 
The tangential Lebesgue-density assumption in Theorem~\ref{thm:GBL} is satisfied by the specific class of practical sets, as defined in \cite{blanchini2008set}, which implies sleekness~\cite{aubin2009viability}, and hence tangential Lebesgue-density by Proposition~\ref{prop:sleek_implies_FDC}. Convex sets are also sleek, and hence tangentially Lebesgue-dense.
\end{remark}

Theorem~\ref{thm:GBL} establishes that robust forward invariance is a compositional property: Global invariance of the interconnected system $\Sigma$ is equivalent to local invariance of each subsystem $\Sigma_i$, $i \in I$, and can therefore be verified independently for each subsystem without knowledge of the full network state. In particular, the equivalence (a)$\Leftrightarrow$(c) shows that checking global invariance of $K$ reduces to checking local invariance of each $K_i$ under the coupling inputs $\mathcal{W}^2_i = \prod_{j \in \mathcal{N}(i)} h_j(K_j)$ generated by the neighboring invariant sets. Furthermore, the equivalence (c)$\Leftrightarrow$(d) reduces local invariance verification to the tangent cone inclusion
\[
f_i(x_i,\mathcal{W}^1_i,\mathcal{W}^2_i)\subseteq T_{K_i}(x_i),
\] 
which depends only on the local dynamics of $\Sigma_i$ and the geometry of $K_i$ at each boundary point $x_i \in \partial K_i$.

\section{Numerical Example}
\label{sec:numerical}

We illustrate the theoretical results on a DC microgrid composed of $N$ distributed generation units (DGUs), each equipped with a Buck converter, an LC filter, and a local droop-based primary controller, interconnected through resistive power lines. This model is standard in the literature~\cite{Tucci2018,Nasirian2015} and represents a class of large-scale networked control systems for which compositional invariance verification is practically essential.

\subsection{System model}

Each DGU $i$ consists of a generic renewable energy source (modelled as a battery) connected to a local load at the point of common coupling through a Buck converter and an LC filter. The primary control layer consists of a droop controller~\cite{Nasirian2015}, which is a decentralized static feedback law that regulates the point of common coupling (PCC) voltage $V_i$ by adjusting the current reference $I_i^{\mathrm{ref}}$ injected by the Buck converter proportionally to the deviation of $V_i$ from the nominal voltage $V_i^{\mathrm{nom}}$:
\begin{equation}
\label{eq:droop}
I_i^{\mathrm{ref}} = I_i^{\mathrm{nom}} - \frac{1}{R_i^d}
\bigl(V_i - V_i^{\mathrm{nom}}\bigr),
\end{equation}
where $R_i^d > 0$ is the droop resistance and $I_i^{\mathrm{nom}}$ is the nominal current setpoint. The droop resistance governs the trade-off between voltage regulation accuracy and current sharing among DGUs: a smaller $R_i^d$ improves voltage regulation at the cost of reduced current sharing flexibility~\cite{Nasirian2015}. Under the quasi-stationary line approximation~\cite{Tucci2018}, which neglects line inductances since their time constants are much smaller than those of the converters, and after substituting~\eqref{eq:droop}, the closed-loop dynamics of DGU $i$ reduce to the scalar first-order system
\begin{equation}
\label{eq:DGU}
C_i \dot{V}_i = -\frac{V_i - V_i^{\mathrm{nom}}}{R_i^d} + I_i^{\mathrm{nom}}
- I_i^L(t) - \sum_{j \in \mathcal{N}(i)} \frac{V_i - V_j}{R_{ij}},
\end{equation}
where $V_i \in \mathbb{R}$ is the PCC voltage (state), $I_i^L(t) \in \mathbb{R}$ is the local load current (external disturbance), $V_j$ are the PCC voltages of the neighbouring DGUs (internal inputs), $C_i > 0$ is the filter capacitance, and $R_{ij} > 0$ is the resistance of the power line connecting DGU $i$ to DGU $j$.

Equation~\eqref{eq:DGU} fits exactly the subsystem framework of Definition~\ref{def:agent} with $x_i = V_i \in X_i \subset \mathbb{R}$, $w_i^1 = I_i^L \in W_i^1 \subset \mathbb{R}$, $w_i^2 = \col(V_j)_{j \in \mathcal{N}(i)} \in W_i^2$, and $y_i = h_i(x_i) = V_i$. The map $f_i$, defined as a single-valued map in this example, is given by the right-hand side of~\eqref{eq:DGU} and is affine in all its arguments, hence Lipschitz with compact values, and Assumption~\ref{ass:SA} is satisfied. The physical parameters are taken from~\cite{Tucci2018} and are given by $C_i = 2.2 \times 10^{-3}$ F, $R_i^d = 0.1 \Omega$, $R_{ij} = 0.5\ \Omega$, $V_i^{\mathrm{nom}} = 48$ V, and $I_i^{\mathrm{nom}} = 5$ A, for all $i \in I$.

The load current $I_i^L(t)$ is modelled as a piecewise-constant signal taking values in $\mathcal{W}_i^1 = [0, 12]$ A, composed of a common slow step shared by all DGUs (alternating every $5$ s between $\approx 1.5$ A and $\approx 10.5$ A), an individual fixed offset per DGU drawn uniformly from $[-1.5, 1.5]$ A, affected by a 5$\%$ noise. 

\subsection{Invariant sets and conditions verification}

The voltage safety requirement at each DGU is that the PCC voltage remains within the interval
\[
K_i = [\underline{V}_i,\, \overline{V}_i] = [47.2,\, 48.6]\ \text{V},
\]
corresponding to a band of approximately $\pm 1.5\%$ around the nominal voltage $V_i^{\mathrm{nom}} = 48$ V. The global safe set is $K = \prod_{i=1}^N K_i \subset \mathbb{R}^N$. Since each $K_i$ is a closed interval, it is convex, hence sleek, and therefore tangentially Lebesgue-dense at every boundary point in the sense of Definition~\ref{def:TLD}, by Proposition \ref{prop:sleek_implies_FDC}. The assumptions of Theorem~\ref{thm:GBL} are therefore satisfied. By Theorem~\ref{thm:GBL}, global robust forward invariance of $K$ for the interconnected microgrid $\Sigma$ under $\mathcal{W}^1 = \prod_{i=1}^N [0,12]$ is equivalent to local robust forward invariance of each $K_i$ for subsystem $\Sigma_i$ under $(\mathcal{W}_i^1, \mathcal{W}_i^2)$, which in turn is equivalent to condition~(d): for all $i \in I$ and all $V_i \in \partial K_i = \{\underline{V}_i, \overline{V}_i\}$,
\[
f_i(V_i,\, \mathcal{W}_i^1,\, \mathcal{W}_i^2) \subseteq T_{K_i}(V_i).
\]
Since $K_i$ is an interval, $T_{K_i}(\underline{V}_i) = \mathbb{R}_{\geq 0}$ and $T_{K_i}(\overline{V}_i) = \mathbb{R}_{\leq 0}$. The coupling constraint set is $\mathcal{W}_i^2 = \prod_{j \in \mathcal{N}(i)} h_j(K_j) = \prod_{j \in \mathcal{N}(i)} [\underline{V}_j, \overline{V}_j]$. A straightforward computation shows that condition (d) of Theorem~\ref{thm:GBL} is satisfied for every DGU $i \in I$. 

\begin{remark}
Let us mention that Condition~(d) of Theorem~\ref{thm:GBL} requires only $2N$ scalar checks, one pair per DGU, each involving only local parameters and the safety bounds of direct neighbours, making the total complexity linear in $N$. By contrast, condition~(b) requires verifying $F(V, w^1) \subseteq T_K(V)$ for all $V \in \partial K$, which is an uncountably infinite set in $\mathbb{R}^N$ and admits no finite verification procedure in general. Related works~\cite{Zonetti2019NECSYS,Zonetti2019ECC} establish compositional invariance results for DC microgrids, but only in the sufficient direction: local invariance implies global invariance, corresponding to (c)$\Rightarrow$(a) in Theorem~\ref{thm:GBL}. The reverse implication is not established in either reference, whereas Theorem~\ref{thm:GBL} provides the first if-and-only-if equivalence for continuous-time interconnected systems.
\end{remark}

\subsection{Simulation results}

We simulate a network configuration with $N = 100$ DGUs with a \emph{$k$-nearest-neighbour topology} with $k = 6$, where each DGU is connected to its three left and three right neighbours on a ring; this is a denser and more realistic interconnection structure that tests the robustness of the invariance certificate to richer coupling patterns. The initial conditions are spread uniformly across $K_i$ as
\[
V_i(0) = \underline{V}_i + \frac{i-1}{N-1}\bigl(\overline{V}_i -
\underline{V}_i\bigr), \quad i = 1, \ldots, N,
\]
so that DGU $1$ starts exactly at the lower boundary $\underline{V}_i = 47.2$ V and DGU $N$ starts exactly at the upper boundary $\overline{V}_i = 48.6$ V. These initial conditions lie on $\partial K$ and Theorem~\ref{thm:GBL} guarantees that no trajectory can escape even when started at the extremes of the safe set. Figures~\ref{fig:N100_volt} and~\ref{fig:N100_dist} show the voltage trajectories and disturbance signals for $N = 100$. One can see that all trajectories remain strictly inside $K_i = [47.2, 48.6]$ V for all $t \in [0, 20]$ s, despite the persistent, highly irregular, and unpredictable load disturbances. The voltage trajectories exhibit two distinct regimes corresponding to the two levels of the common disturbance step. When $I_i^L(t) \approx 10.5$ A, the network equilibrium lies near $V_i^* = V_i^{\mathrm{nom}} + R_i^d(I_i^{\mathrm{nom}} - 10.5) = 47.35$ V, causing all voltages to approach the lower boundary. When $I_i^L(t) \approx 1.5$ A, the equilibrium lies near $V_i^* = 48.35$ V, and voltages approach the upper boundary.

The $N = 100$ simulation further illustrates the scalability of the approach. Condition~(d) requires only $2 \times 100 = 200$ scalar inequality checks. Adding a new DGU to the network requires checking only two additional inequalities for that DGU and does not require recomputing or modifying the certificates of any existing DGU, which is the plug-and-play property enabled by the compositional structure of Theorem~\ref{thm:GBL}.

\begin{figure}[t]
    \centering
    \includegraphics[width=\columnwidth]{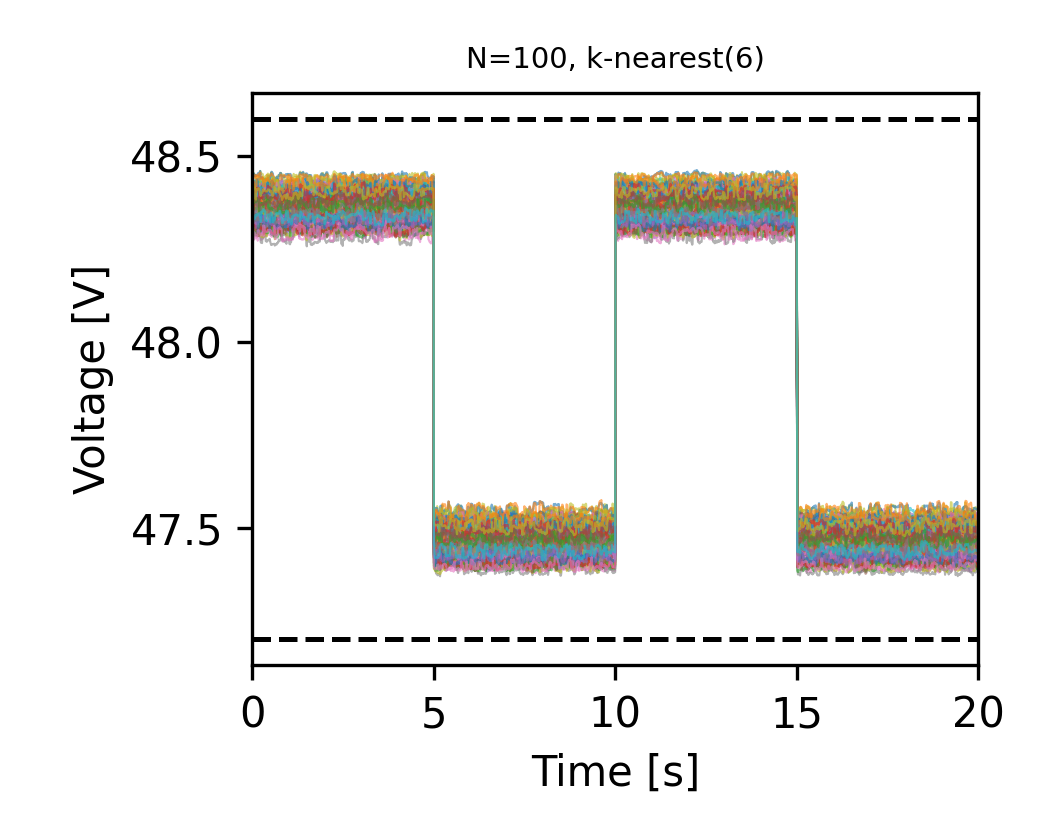}
    \caption{Voltage trajectories $V_i(t)$ for $N=100$ DGUs with
    $k$-nearest-neighbour topology ($k=6$). All $100$ trajectories remain
    inside $K_i = [47.2, 48.6]$ V for all $t \in [0,20]$ s. Condition~(d)
    is verified by $2 \times 100 = 200$ scalar inequalities, independently of
    the network topology.}
    \label{fig:N100_volt}
\end{figure}

\begin{figure}[t]
    \centering
    \includegraphics[width=\columnwidth]{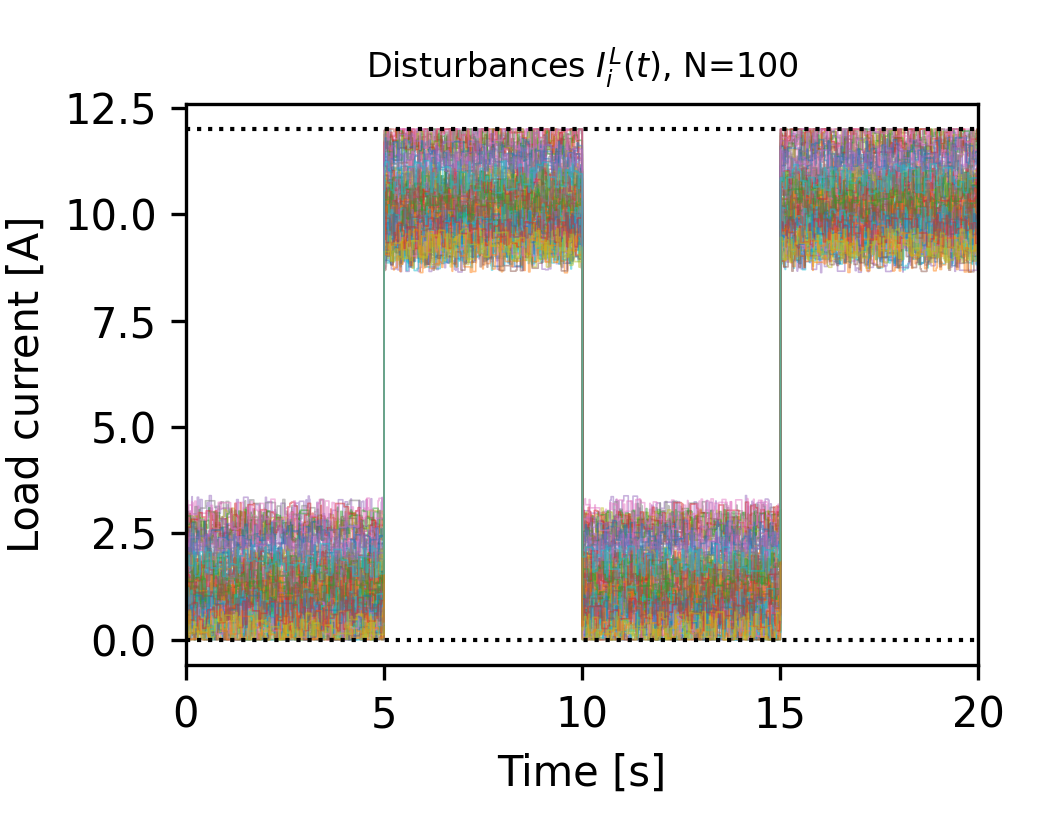}
    \caption{Piecewise-constant load disturbances $I_i^L(t) \in [0, 12]$ A
    for $N=100$ DGUs.}
    \label{fig:N100_dist}
\end{figure}

\section{Conclusion}

This paper established the first if-and-only-if compositional 
invariance result for continuous-time interconnected systems: 
robust forward invariance of the global product set 
$K = \prod_{i \in I} K_i$ for the interconnected system is 
equivalent to robust forward invariance of each local set $K_i$ 
for the corresponding subsystem under coupling inputs from 
neighboring invariant sets. The key geometric enabler is the 
new notion of tangential Lebesgue-density, which is strictly 
weaker than classical sleekness yet sufficient to ensure that 
the product of local tangent cones equals the tangent cone of 
the product set. Scalability is demonstrated on a DC microgrid where verification complexity grows 
linearly in the number of subsystems.

Two natural extensions are left for future work. First, 
the present result assumes that the system evolves on the 
whole state space $\mathbb{R}^n$. Extending the compositional 
framework to differential inclusions defined on constraint sets, as studied 
in~\cite{Reynaud2025} would broaden the class of 
systems covered. Second, while the present result is stated for uncontrolled systems, an important 
practical extension is the co-design of decentralized 
controllers that simultaneously render the local sets invariant.

\bibliographystyle{ieeetr}
\bibliography{bibliography} 

\end{document}